\documentclass{amsart}
\usepackage{amsmath, latexsym, amsfonts, stmaryrd, amssymb, amsthm, amscd, array, caption, enumerate}
\usepackage[T1]{fontenc}
\usepackage[toc,page]{appendix}
\usepackage[utf8]{inputenc}
\usepackage{graphics,epsf,psfrag,epsfig,color}
\usepackage[colorlinks=true, pdfstartview=FitV, linkcolor=blue, citecolor=blue, urlcolor=blue,pagebackref=false]{hyperref}

\usepackage[showlabels,sections,floats,textmath,displaymath]{}

\numberwithin{equation}{section}
\newtheorem{theorem}{Theorem}[section]
\newtheorem{lemma}[theorem]{Lemma}
\newtheorem{proposition}[theorem]{Proposition}

\renewenvironment{proof}[1][Proof]{\begin{trivlist}
\item[\hskip \labelsep {\bfseries #1}]}{\qed\end{trivlist}}

\newcommand{\ind}{\mathbf{1}}
\renewcommand{\ge}{\geq}
\renewcommand{\le}{\leq}

\renewcommand{\tilde}{\widetilde}
\renewcommand{\hat}{\widehat}

\DeclareMathSymbol{\leqslant}{\mathalpha}{AMSa}{"36} 
\DeclareMathSymbol{\geqslant}{\mathalpha}{AMSa}{"3E} 
\DeclareMathSymbol{\eset}{\mathalpha}{AMSb}{"3F}     
\renewcommand{\leq}{\;\leqslant\;}                   
\renewcommand{\geq}{\;\geqslant\;}                   
\newcommand{\dd}{\,\text{\rm d}}             

\newcommand{\sumtwo}[2]{\sum_{\substack{#1 \\ #2}}} 
\newcommand{\prodtwo}[2]{\prod_{\substack{#1 \\ #2}}}     

\renewcommand{\u}[1]{\underline{#1}}

\newcommand{\cA}{{\ensuremath{\mathcal A}} }

\newcommand{\cE}{{\ensuremath{\mathcal E}} }

\newcommand{\cT}{{\ensuremath{\mathcal T}} }

\newcommand{\cI}{{\ensuremath{\mathcal I}} }

\newcommand{\cY}{{\ensuremath{\mathcal Y}} }
\newcommand{\cS}{{\ensuremath{\mathcal S}} }

\newcommand{\bP}{{\ensuremath{\mathbf P}} }
\newcommand{\bE}{{\ensuremath{\mathbf E}} }

\newcommand{\bbE}{{\ensuremath{\mathbb E}} }

\newcommand{\bbL}{{\ensuremath{\mathbb L}} }

\newcommand{\bbN}{{\ensuremath{\mathbb N}} }

\newcommand{\bbP}{{\ensuremath{\mathbb P}} }

\newcommand{\bbR}{{\ensuremath{\mathbb R}} }

\newcommand{\bbZ}{{\ensuremath{\mathbb Z}} }

\newcommand{\ga}{\alpha}
\newcommand{\gb}{\beta}
\newcommand{\gd}{\delta}
\newcommand{\gep}{\varepsilon}       

\newcommand{\gD}{\Delta}

\newcommand{\go}{\omega}

\newcommand{\gl}{\lambda}
\newcommand{\gL}{\Lambda}

\makeatletter
\def\captionfont@{\footnotesize}
\def\captionheadfont@{\scshape}

\long\def\@makecaption#1#2{%
  \vspace{2mm}
  \setbox\@tempboxa\vbox{\color@setgroup
    \advance\hsize-6pc\noindent
    \captionfont@\captionheadfont@#1\@xp\@ifnotempty\@xp
        {\@cdr#2\@nil}{.\captionfont@\upshape\enspace#2}%
    \unskip\kern-6pc\par
    \global\setbox\@ne\lastbox\color@endgroup}%
  \ifhbox\@ne 
    \setbox\@ne\hbox{\unhbox\@ne\unskip\unskip\unpenalty\unkern}%
  \fi
  \ifdim\wd\@tempboxa=\z@ 
    \setbox\@ne\hbox to\columnwidth{\hss\kern-6pc\box\@ne\hss}%
  \else 
    \setbox\@ne\vbox{\unvbox\@tempboxa\parskip\z@skip
        \noindent\unhbox\@ne\advance\hsize-6pc\par}%
\fi
  \ifnum\@tempcnta<64 
    \addvspace\abovecaptionskip
    \moveright 3pc\box\@ne
  \else 
    \moveright 3pc\box\@ne
    \nobreak
    \vskip\belowcaptionskip
  \fi
\relax
}
\makeatother
\def\writefig#1 #2 #3 {\rlap{\kern #1 truecm
\raise #2 truecm \hbox{#3}}}

\newcommand{\tf}{\mathtt{F}}

\newcommand{\Var}{\mathbb{V}\mathrm{ar}}

\renewcommand{\subset}{\subseteq}

\renewcommand{\phi}{\varphi}

\begin{document}

\title[Free energy of the directed polymer in dimension $1+2$]{The High-temperature behavior for the directed polymer in dimension $1+2$}

\author{Quentin Berger}
\address{LPMA, Universit\'e Pierre et Marie Curie\\
Campus Jussieu, case 188\\
4 place Jussieu, 75252 Paris Cedex 5, France}
\email{quentin.berger@upmc.fr}

\author{Hubert Lacoin}
\address{IMPA,
Estrada Dona Castorina 110,
Rio de Janeiro / Brasil 22460-320}
\email{lacoin@impa.br}

\begin{abstract}
We investigate the high-temperature behavior of the directed polymer model in dimension $1+2$. More precisely we study the difference $\gD \tf(\gb)$ between the quenched and annealed free energies for small values of the inverse temperature $\gb$.
This quantity is associated to localization properties of the polymer trajectories, and is related to the overlap fraction of two replicas.
Adapting recent techniques developed by the authors in the context of the disordered pinning model \cite{cf:BL15}, 
we identify the sharp asymptotic high temperature behavior 
\[\lim_{\gb\to 0} \, \gb^2 \log \gD \tf(\gb) = -\pi \, .\]
 \\[5pt]
  \\[5pt]
\end{abstract}

\maketitle

\section{Introduction}

The directed polymer model has been introduced by Huse and Henley (in dimension $1+1$) in 1985 \cite{cf:HH}
as an effective model for an interface in the Ising model with impurities. 
It was shortly afterwards generalized to arbitrary dimension $1+d$, where it stands as a model for a stretched polymer interacting with an inhomogeneous solvent. The behavior of the polymer trajectory  depends very on value of $d$, see \cite{cf:CSY_review} for a review.

\medskip

In dimension $1+3$ and higher there is a phase transition between a high temperature diffusive phase 
 for which there is a Brownian scaling \cite{cf:B,cf:CY}, and a localized phase where the polymer tends to pin on 
 a few narrow corridors were the environment is more favorable (see \cite{cf:CH, cf:CSY} for rigorous evidence of the phenomenon).
 
\medskip

In dimension $1+1$, the polymer is localized at every temperature. Moreover
it belongs to the KPZ universality class, which has been the object of intense studies in the recent year (for a connection between directed polymer and the KPZ equation, see e.g.~\cite{cf:AKQ} and references therein).

\medskip

The dimension $1+2$ is critical for the model. 
It is known that localization occurs at all temperature (see \cite{cf:CSY,cf:L}), but the difference between the quenched and annealed free-energy of the system, which is a quantitative indicator of localization, grows very slowly with the inverse-temperature.
The aim of this paper is to obtain sharp information on the asymptotic behavior of this free-energy difference.

\subsection{Directed polymer in random environment}

We let $\bP$ denote the law of $S=(S_n)_{n\geq 0}$ the symmetric nearest-neighbor random walk on $\bbZ^d$, starting from $0$, 
and whose increment are i.i.d.\ with law 
\begin{equation}
\bP(S_1=x)=\frac{1}{2d}\ind_{\{|x|=1\}},
\end{equation}
where $|\cdot|$ is the $l_1$ norm.

Let $\go=(\go_{i,x})_{i\geq0, x\in\bbZ^d}$ be a field of i.i.d.\ random variables with law $\bbP$, which are centered and have unit variance, $\bbE[\go_{i,x}]=0$ and $\bbE[(\go_{i,x})^2]=1$.
We also assume that they have a finite exponential moment in a neighborhood of zero, i.e. that for some positive $c$,
\begin{equation}\label{expomom}
 \forall \gb\in [-c,c], \quad \gl(\gb):=\log \bbE[e^{\gl \go_{i,x}}]<+\infty.
\end{equation}

\medskip

Given the random environment $\go$ and the inverse temperature $\gb>0$, we define the following Gibbs transformation of the law $\bP$ of the random walk up to length $N$
\begin{equation}
\frac{\dd \bP_{N}^{\gb,\go}}{\dd \bP} (S) := \frac{1}{Z_{N}^{\gb,\go}}\, \exp\left( \sum_{n=1}^N \gb\go_{n,S_n}\right) \, ,
\end{equation}
 where $Z_{N}^{\gb,\go}$ is the partition function
\begin{equation}
Z_{N}^{\gb,\go} = \bE \left[ \exp\left(  \sum_{n=1}^N \gb\go_{n,S_n}\right)  \right] \, .
\end{equation}
The free energy (or pressure) of the system is defined by 
\begin{equation}
\tf (\gb) := \lim_{N\to\infty} \frac1 N \log Z_{N}^{\gb,\go}.
\end{equation}
The limit is known to exist and be $\bbP$-a.s.\ constant, see \cite[Prop.~2.5]{cf:CSY}.
It is also known that the converge holds in $\bbL_1$ and hence that 
\begin{equation}
\tf (\gb) =\lim_{N\to\infty} \frac1 N \bbE \left[ \log Z_{N}^{\gb,\go}\right].
\end{equation}
An easy upper-bound on $\tf(\gb)$ it is given by Jensen's inequality 
\begin{equation}
\tf (\gb) \le \lim_{N\to\infty} \frac1 N \log \bbE\left[ Z_{N}^{\gb,\go} \right]= \gl(\gb).
\end{equation}
We refer to this upper bound as the \textit{annealed} free-energy while $\tf(\gb)$ is the \textit{quenched} one.
Knowing whether or not this inequality is sharp gives information on the localization of the trajectory.
Heuristically  $\tf (\gb)<\gl(\gb)$ corresponds to localization of the trajectories under $\bP_{N}^{\gb,\go}$ 
around favorite corridors where $\go$ is favorable, whereas
$\tf (\gb)=\gl(\gb)$ implies diffusivity of $S$. 
This has been largely put on rigorous ground both for the diffusive case \cite{cf:B,cf:CY} and the localized one \cite{cf:CH,cf:CSY}.
Moreover, it is known that $\tf (\gb)=\gl(\gb)$ for small value of $\gb$ when $d\ge 3$ \cite{cf:B} while the inequality is always strict for $d=1$ \cite{cf:CV} and $d=2$ \cite{cf:L}.

\medskip

When $\tf (\gb)<\gl(\gb)$, the difference $\Delta \tf(\gb)= \gl(\gb)-\tf(\gb)>0$ gives some indication on how much localized the trajectories are: Carmona and Hu \cite{cf:CH} (and later Comets Shiga and Yoshida)gave an explicit link between $\Delta \tf(\gb)$ and the overlap fraction of two replicas, namely in our context
\[\gD\tf(\gb) = \gl(\gb)\lim_{N\to\infty} \frac{1}{N} \sum_{k=1}^N (\bP_{k-1}^{\gb,\go} )^{\otimes 2}(S_k^{(1)}=S_k^{(2)}) \quad \bbP-a.s. \]

\smallskip
In dimension $1$, it is known that $\Delta \tf(\gb)$ scales like $\gb^4$ (see \cite{cf:L, cf:Watbled, cf:AY}), and it is conjectured \cite{cf:SS1, cf:SS2} that 
\begin{equation}\label{cf:conj}
\lim_{\gb\to 0} \gb^{-4} \Delta \tf(\gb)= \frac{1}{24} \, .
\end{equation}
The exponent $4$ is very much related to the $\gb=N^{-1/4}$ scaling which is required to obtain a non-trivial intermediate disorder regime limit, see \cite{cf:AKQ}.

\medskip

\subsection{Main result}

In this paper we focus on the case of $d=2$, the critical dimension for directed polymers, for which 
the renormalized free energy $\Delta \tf(\gb)$ vanishes faster than any power of $\gb$.
In \cite{cf:L}, the author showed the existence of a constant $c$ such that for $\gb\leq 1$,
\[ -c^{-1} \gb^{-4} \leq \log \gD \tf(\gb) \leq -c \, \gb^{-2}.\]
In \cite{cf:Naka}, the lower bound was improved to $\log \gD \tf(\gb) \geq - c_{\gep} \gb^{-(2+\gep)}$ for any $\gep>0$.

\medskip

Our main theorem  significantly improves over previous results and identifies the sharp critical behavior of $\Delta \tf(\gb)$.

\begin{theorem}
\label{thm:main}
For $d=2$,
\[\lim_{\gb\to 0}\, \gb^2 \log \Delta \tf(\gb) = - \pi \, .\]
\end{theorem}

\subsection{Strategy of the proof and organization of the paper}
\label{sec:aboutproof}

Our result improves both existing lower and upper bound on $\gD \tf(\gb)$.
The main part of the work concerns the lower bound. 

To derive it we us a by now standard fractional moment/coarse-graining procedure, 
employed in the context of pinning models \cite{cf:DGLT,cf:GLT11} recently enhanced in \cite{cf:BL15}, and adapted for the directed polymer model 
in \cite{cf:L,cf:Naka}. 
Here, we relie in particular on new ideas that have been introduced in \cite{cf:BL15} to obtain optimal bounds on the critical point shift in disordered pinning.
Let us sketch briefly how the different parts of the proof articulate.

\medskip

First we realize that to control the free-energy it is sufficient to have a control on $\bbE[\sqrt{\hat Z^{\gb,\go}_N}]$
which is easier to work with than the $\log$ partition function.
Then, to obtain the desired upper bound, we proceed in three steps which we introduce here in a rather informal manner:
\begin{itemize}
\item [(i)] We perform a coarse-graining of the system, dividing it into cells of length $\ell$ and width $\sqrt{\ell}$ 
(to fit with the random walk diffusive scaling) where $\ell$ depends on $\gb$ and gets very large when $\gb$ gets small.
We choose $\ell$ to be roughly the inverse of $\gD\tf(\gb)$ or rather the inverse of the bound we would like to prove for it.
The idea behind this procedure is to ``factorize'' the partition function of a system of size much larger than $\ell$
and isolate the contribution of each cell.
Then if one is able to show that partition function ``restricted to a cell'' 
is small, we want to use the factorization procedure to deduce a bound on the free energy.

\item [(ii)] The coarse grained trajectory is defined as the projection of the original trajectory $S$ on this rougher lattice 
(we give a more proper definition in the core of the paper).
We decompose the partition function of a system whose size is a multiple of $\ell$ 
by isolating the contribution of each coarse grained trajectory.
By using the inequality $\sqrt{\sum a_i} \leq \sum \sqrt{a_i}$ valid for any collection of positive $a_i$’s,  
we reduce ourselves to estimate the square root moment of partition functions restricted to a single coarse grained trajectory.

\item [(iii)] We estimate these square root contribution of coarse grained trajectories by performing a ``change of measure'' 
which makes the environment $\go$ less favorable in the visited cells.
The way we choose this change of measure is quite elaborate and is based on a multilinear form of the $\go_{n,x}$ in the cell. 
It is described in details in Section \ref{sec:chgmeas}.
\end{itemize}

The steps (i) and (ii) are identical to those performed in \cite{cf:L} and are presented in Section \ref{sec:lower},
however the change of measure is significantly improved with respect to that of \cite{cf:L} and builds on the innovations introduced in \cite{cf:BL15}. 
In Section \ref{sec:keylemma}, we prove the technical estimates needed to control the effect of the change of measure procedure.

\medskip

The upper bound is obtained in Section \ref{sec:upper} thanks to an estimate on the second moment of the partition function, 
together with a concentration argument of $\log Z_{N}^{\gb,\go}$ around its mean inspired by \cite{cf:CTT}.

\subsection{Generalization of the result?}

The techniques described in Section \ref{sec:aboutproof} could be adapted to a more general context. Indeed, one might consider the model in which the random walk $S$ is not the simple symmetric random walk on $\bbZ^d$, but belongs to the domain of attraction of an $\ga$-stable law with $\ga\in(0,2]$, see \cite{cf:Comets}. Let us consider the case of the dimension $1+1$: it has been showed that weak disorder holds for $\gb$ small enough when $\ga\in(0,1)$, see \cite{cf:Comets}, and that strong disorder holds for any $\gb>0$ when $\ga \in (1,2)$, see \cite{cf:MTT} (in a continuous setting). A similar question has been studied in \cite{cf:CSZ1}, where a \emph{disordered} scaling limit can be constructed whenever $\ga\in (1,2]$. The case $\ga=1$ is marginal, as it is the case of the simple random walk in dimension $1+2$, and it is likely that it could be dealt with the same methods as presented here: one should be able to obtain a necessary and sufficient condition for the existence of a weak disorder phase (note that this is related to the notion of disorder irrelevance, studied in \cite{cf:BL15}). In general, localization should occur for all $\gb>0$ if and only if $S$ is recurrent, and the growth of the excess free energy $\gD \tf(\gb)$ should be related to that of the mean intersection local time up to time $N$, cf.~\eqref{def:D} (analogously to \cite[Prop.~6.1-7.1]{cf:BL15}).

\subsection{Some notations}

We write  
\begin{equation}
\label{renormalized}
 \hat Z_{N}^{\gb,\go}:=e^{-N\gl(\gb)}  Z_{N}^{\gb,\go}
\end{equation}
for the renormalized partition function.
We introduce the intersection local time up to time $N$,
\begin{equation}
L_N(S^{(1)},S^{(2)}) = \sum_{t=1}^N \ind_{\{ S^{(1)}_t = S^{(2)}_t\}}\, .
\end{equation}
For $t\in\bbN$ and $x\in\bbZ^2$, we write
\[p(t,x) := \bP(S_t=x)\, ,\]
for the kernel of the symmetric simple random walk on $\bbZ^2$. 
A central quantity for the model is the mean intersection local time up to $N$,
\begin{equation}
\label{def:D}
D(N) : = \sum_{t=1}^N \bP(S^{(1)}_t = S^{(2)}_t ) =\sum_{t=1}^N p(2t,0) \ \ \stackrel{N\to\infty}{\sim } \ \ \frac1\pi  \log N \, .
\end{equation}
Note that $D(N)$ can also be written as $\sum\limits_{t=1}^N \sum\limits_{x\in\bbZ^2} p(t,x)^2$.

\section{Lower-bound}
\label{sec:lower}

\subsection{Fractional moment and coarse-graining}

To bound the free energy from above we have to estimate the expected value of $\log Z_{N}^{\gb,\go}$.
Using Jensen's inequality, we can reduce  the problem into having to estimate only the square root, which turns out to be more convenient. 
We have 
\begin{equation}
\bbE\big[ \log \hat Z_{N}^{\gb,\go} \big]= 2 \bbE\big[ \log \sqrt{ \hat Z_{N}^{\gb,\go}} \big]\le 2\log \bbE\big[ \sqrt{ \hat Z_{N}^{\gb,\go}} \big],
\end{equation}
and hence 
\begin{equation}\label{paietaliminf}
\Delta \tf (\gb) \ge  - \liminf_{N\to\infty} \, \frac{2}{N} \log \bbE \big[ \sqrt{ \hat Z_N^{\gb,\go}}\big] \, .
\end{equation}
We choose to split the system into ``cells'' of length $\ell$ which we choose to be equal to
\begin{equation}\label{def:ell}
\ell = \ell_{\gb,\gep} := \exp\left( (1+2\gep)\, \frac{\pi} {\gb^2}  \right) \,  ,
\end{equation}
where $\gep>0$ is a parameter (fixed for the rest of the proof), which we choose to be small.
 The reason for this choice of coarse-graining length will appear later in the proof.
We consider a system whose length is a multiple of $\ell$: $N=m\ell$, $m\in\bbN$.
For every $y\in\bbZ^2$, we define a window centered at $y\sqrt{\ell}$ (we assume for simplicity that 
$\sqrt{\ell}$ is an even integer), and of width $\sqrt{\ell}$:
\[\Lambda_y:= y\sqrt{\ell} + (-\tfrac12 \sqrt{\ell},\tfrac12 \sqrt{\ell}]^2\cap \bbZ^2 \, ,\]
Note that $\gL_y$ contains $\ell$ points. 
Given any $\cY = (y_1,\ldots, y_m) \in (\bbZ^2)^m$, we define the event
\begin{equation}
\cE_{\cY} : = \big\{  \forall i\in\{1,\ldots, m\}  , \, S_{i\ell} \in \Lambda_{y_i}  \big\}\, .
\end{equation}
If $S \in\cE_\cY$, then $\cY$ is a coarse-grained version of the trajectory of $S$.
The width of the cells is chosen to match the scaling of the random-walk.

\medskip

We decompose  $\hat Z_{N}^{\gb,\go}$ according to the  contribution of the different coarse-grained trajectories
\begin{equation}
\label{eq:fractional}
\hat Z_N^{\gb,\go} = \hat Z_{m\ell}^{\gb,\go} =
\sum_{\cY  \in (\bbZ^2)^m} \bE \left[\exp\left(\sum_{n=1}^N\gb\go_{n,S_n}-\gl(\gb)\right)\ind_{\{S \in \cE_{\cY} \}}\right]=:
 \sum_{\cY  \in (\bbZ^2)^m} Z_{\cY}.
\end{equation}
Using the inequality $(\sum_{i\in \cI}  a_i)^{1/2}\leq \sum_{i\in \cI} a_i^{1/2}$, valid for any family of non-negative $a_i$'s, we obtain
\begin{equation}\label{paietonfrac}
\bbE\left[ (\hat Z_N^{\gb,\go} )^{1/2}\right] \leq \sum_{\cY  \in (\bbZ^2)^m} \bbE \left[ (Z_{\cY})^{1/2} \right]\, .
\end{equation}
We are therefore left with estimating $\bbE \left[ (Z_{\cY})^{1/2} \right]$ for every coarse-grained trajectory $\cY$.
As a consequence of \eqref{paietaliminf} and \eqref{paietonfrac} we have 
\begin{equation}
\gD\tf (\gb)\ge - \liminf_{m\to \infty} \frac{2}{m \ell} \log \left( \sum_{\cY  \in (\bbZ^2)^m} \bbE \left[ (Z_{\cY})^{1/2} \right] \right),
\end{equation}
We obtain the lower bound in Theorem \ref{thm:main} as a consequence of the following result
\begin{proposition}
\label{prop:fracmoment}
For any $\gep>0$, there exists some $\gb_\gep>0$ such that, for every $\gb\in(0, \gb_\gep)$, and $m\ge 1$
\[ \sum_{\cY  \in (\bbZ^2)^m} \bbE \left[ (Z_{\cY})^{1/2} \right] \leq 2^{-m} \, .\]
\end{proposition}
This statement implies indeed  that $\Delta \tf(\gb)\ge (2\log 2)\, \ell^{-1}$,
and thus from the definition of $\ell$ \eqref{def:ell}, for any arbitrary $\gep>0$, we have
\begin{equation}
\liminf_{\gb\to \infty} \gb^2 \log \Delta  \tf(\gb)\ge - (1+2\gep) \pi.
\end{equation}

\subsection{The change of measure argument}

Let $g_{\cY}(\go)$  be a positive function, that can be interpreted as a probability density if renormalized. Using the  Cauchy-Schwarz inequality, we have 
\begin{equation}
\label{eq:CS}
\left( \bbE \left[ (Z_{\cY})^{1/2} \right]\right)^2 \leq \bbE \big[ g_{\cY}(\go)^{-1} \big] \, \bbE \big[ g_{\cY}(\go)\, Z_{\cY} \big] \, .
\end{equation}
The idea is then to choose $g_{\cY}(\go)$ such that $\bbE \big[ g_{\cY}(\go)^{-1} \big]$ is not much larger than one, 
but that lowers significantly the expectation of $Z_{\cY}$. Hence we want $g_{\cY}$ to be of order $1$ for ``typical environments'', 
but much smaller for atypical $\go$ which results in high values of  $Z_{\cY}$ 
(the underlying idea being that these are the ones who carry the mass in the expectation). 

\medskip

As we want to affect the partition function restricted to paths in $\cE_{\cY}$, we choose a change of measure $g_{\cY}(\go)$ which only affects 
the environment in a corridor which is centered on the location of the paths.
To make certain that most trajectories in $\cE_{\cY}$ are affected by the change, we apply it in a region which is slightly wider than $\sqrt{\ell}$: 
For any $y\in\bbZ^2$, let us define 
\[\tilde \gL_{y}:= y\sqrt{\ell} + [-R\sqrt{\ell},R\sqrt{\ell}]^2\cap \bbZ^2\]
where $R$ is chosen sufficiently large (see the proof of Lemma \ref{lem:EX}). Note that $\tilde \gL_y$ contains $4R^2 \ell$ points.

We choose  $g_{\cY}$ to be a function of $\go$ restricted to $\bigcup_{i=1}^n B_{(i,y_{i-1})}$ for $i=1, \dots, m$ where 
\begin{equation}
B_{(i,y)} :=  [(i-1)\ell+1,i\ell] \times \tilde\gL_{y} \, ,
\end{equation}
Because of our coarse-graining, it is natural that we choose  $g_{\cY}$ as a product of functions of the environment restricted to one cell
$(\go_{n,x})_{(n,x)\in B_{(i,y_{i-1})}}$.

\medskip

We let $X(\go)$ be a function of  $(\go_{n,x})_{(n,x)\in B_{(1,0)}}$ which we specify in the next section and
which satisfies
\begin{equation}
\bbE[X(\go)]=0, \quad \bbE [(X(\go))^2]\leq 1.
\end{equation} 
We define  $X^{(i,y)}$ as the space-time ``translation'' of $X$
\begin{equation}\label{rer}
X^{(i,y)}(\go):= X(\theta^{(i-1)\ell,\sqrt{\ell}y} \go),
\end{equation}
where $\theta^{a,b}$ is the shift operator: $(\theta^{a,b} \go)_{t,x} := \go_{t+a,x+b}$.
Finally, given  $K>0$ which is chosen large enough, we set
\begin{equation}
\label{def:g}
\begin{split}
g_{(i,y)}(\go) & := \exp\Big( - K \, \ind_{\{X^{(i,y)}(\go) \geq e^{K^2}\}}  \Big) \, , \\
\quad  g_{\cY}(\go)  & := \prod_{i=1}^m g_{(i,y_{i-1})}(\go)
\end{split}
\end{equation}
With this definition, and provided that $K$ is large, we have
\[
\bbE \big[ g_{(i,y)}(\go)^{-1} \big] = 1 + (e^{K}-1) \bbP \big( X^{(i,y)}(\go) \geq e^{K^2} \big) \leq 1 + (e^{K}-1)e^{-2 K^2}  \leq 2 \, ,
\]
and hence by independence of the $g_{(i,y_{i-1})}$, $i=1,\dots,m$  (which are functions of $\go$ on disjoint sets by construction), we have
\begin{equation}
\label{eq:CS1}
\bbE \big[ g_{\cY}(\go)^{-1} \big]  \leq 2^{m} \, .
\end{equation}
The main task is then to estimate the effect on $Z_{\cY}$ of the multiplication by $g_{\cY}$. We have
\begin{equation}\label{coucou}
\bbE \big[ g_{\cY}(\go)\, Z_{\cY} \big] = 
\bE\left[ \bbE \left[ g_{\cY}(\go) e^{\sum_{n=1}^N \gb\go_{n,S_n} -\gl(\gb)}\right] \, \ind_{\cE_\cY}\right]
\end{equation}
Note that for a fixed trajectory $S$, the measure $\bbP^S$ on $\go$ defined by 
\begin{equation}
\label{def:PS}
\frac{\dd \bbP^S}{\dd \bbP}(\go) := e^{\sum_{n=1}^N \gb\go_{n,S_n} -\gl(\gb)} \, ,
\end{equation}
is a probability measure.
Under $\bbP^S$, $\go$ is still a field of independent random variables 
(in particular the $g_{(i,y_{i-1})}(\go)$, $i=1,\dots,m$ are still independent), but there are not identically distributed: 
the law of $(\go_{n,S_n})_{1\leq n\leq N}$ has been exponentially tilted. 
The variance and expectation of $\go_{n,x}$ for $1\le n\le N$ are then given by 
\begin{equation}
\label{newExpVar}
\bbE^S[\go_{n,x}] = \gl'(\gb) \ind_{\{S_n=x\}} \, ,\qquad \Var^S(\go_{n,x}) = 1 + (\gl''(\gb)-1) \ind_{\{S_n=x\}} \, 
\end{equation}
where $\gl'$ and $\gl''$ denote the two first derivative of $\gl$.
In what follows we will always choose $\gb$ sufficiently small so that 
\begin{equation}\label{eq:boundsonevar}
 \left|\frac {\gl'(\gb)-\gb}{\gb}\right| \le \gep^3  \quad \text{ and }  \quad \gl''(\gb)\le 1+\frac{\gep^3}{2}.
\end{equation}

With this newly defined measure, the identity \eqref{coucou} can be rewritten as follows
\begin{equation}
\bbE \big[ g_{\cY}(\go)\, Z_{\cY} \big] = 
\bE\left[ \bbE^S \left[ g_{\cY}(\go) \right] \, \ind_{\cE_\cY}\right]=\bE\left[ \prod_{i=1}^m \bbE^S \left[ g_{(i,y_{i-1})}(\go) \right] \, \ind_{\cE_\cY}\right].
\end{equation}
Using the product structure of $g_{\cY}(\go) := \prod_{i=1}^m g_{(i,y_{i-1})}(\go)$, 
we perform an approximate factorization of the above expression by considering the worse possible intermediate points for $S$. It yields the following upper bound
\begin{equation}
\label{reduce1block}
\prod_{i=1}^m \max_{x\in\gL_{y_{i-1}}} \bE\left[  \bbE^S \left[ g_{(i,y_{i-1})}(\go) \right] \, ; \, S_{i\ell} \in\Lambda_{y_i} \big| \, S_{(i-1)\ell}=x\right]  \notag
\end{equation}
Using translation invariance \eqref{rer} and summing over all $\cY$ we have
\begin{equation}
 \sum_{\cY  \in (\bbZ^2)^m} \bbE \big[ g_{\cY}(\go)\, Z_{\cY} \big]^{1/2}  \leq \left( \sum_{y\in\bbZ^2}  
 \max_{x\in\gL_{0}} \left( \bE_x\left[  \bbE^S \left[ g_{(1,0)}(\go) \right] \, ; \, S_{\ell} \in\Lambda_{y} \right]\right)^{1/2}  \right)^m,
\end{equation}
where $\bP_x$ denotes the law of the simple random walk starting from $x$.
Therefore, one only needs to consider one block: combining this with Lemma \ref{lem:key} and \eqref{eq:CS},\eqref{eq:CS1}, this proves Proposition \ref{prop:fracmoment}.

\subsection{Choice of the change of measure}
\label{sec:chgmeas}

We now specify our choice of $X$.
With the expression that we have chosen for $g$, we want $X$ to be typically larger than $K$ under $\bbE^S$,
at least for most realizations of $S$.
We choose $X$ to be a positive $q$-linear form of $(\go_{n,x})_{(n,x)\in B_{(1,0)}}$, which corresponds more or less to 
the term of order $q$ appearing in the Taylor expansion of the partition function $Z^{\gb,\go}_{\ell}$.
We set
\begin{equation}\label{eq:defq}
q_{\ell} : = (\log\log \ell)^2 \, .
\end{equation}
To simplify the calculations, we also reduce the interactions (in time) to a range $u \ll \ell$. We choose $u=u_{\ell}:=\lfloor \ell^{1-\gep^2}\rfloor$.
Note that this gives (cf. \eqref{def:D}) 
\begin{equation}\label{eq:frameu}
D(u)\stackrel{\gb\to 0}{\sim} \frac{1-\gep^2}{\pi} \log \ell,
\end{equation}
so that the definition of $\ell$ ensures that for $\gb$ sufficiently small (and if $\gep<1/10$)
\begin{equation}\label{eq:compar}
(1+\gep) \le \gb^2 D(u)\le (1+2\gep).
\end{equation}
We introduce the set of increasing sequences with increments no larger than $u$
\begin{equation}
J_{\ell,u}:= \{ \u{t}:=(t_0,\ldots, t_q) \in \bbN^{q+1} \, | \, 1\leq  t_0 < \cdots < t_q \leq \ell \, ; \, (t_{j}-t_{j-1}) \leq u,  \, \forall j\in\{1,\ldots, q\}\} \, ,
\end{equation}
We now define
\begin{equation}
\label{def:X}
X(\go) : = \frac{1}{2 R \ell D(u)^{q/2}}   \sum_{ \u{x} \in (\tilde\gL_0)^{q+1} ,\, \u{t}\in J_{\ell,u} } P(\u{t},\u{x})\,  \go_{\u{t},\u{x}}
\end{equation}
where for any $\u{x}=(x_0,\ldots, x_q) \in (\tilde\gL_0)^{q+1} $ we set $\go_{\u{t},\u{x}} = \prod_{j=0}^q \go_{t_j,x_{j}}$, and 
\begin{multline}
P(\u{t},\u{x})=\prod_{j=1}^q p(t_j - t_{j-1},x_j -x_{j-1})\ind_{\{|x_j -x_{j-1}|\le \rho(t_j - t_{j-1})  \}} \\ = 
\bP_{x_0} \big( S_{t_j-t_0} = x_j-x_0 \, , \, \forall j\in \{1,\ldots, q\} \big)\,  \ind_{\{|x_j -x_{j-1}|\le \rho(t_j - t_{j-1}) \, , \, \forall j\in \{1,\ldots, q\}\}}.
\end{multline}
Here and later in the proof $|x|=|x_1|+|x_2|$ denotes the $l_1$ norm on $\bbZ^2$, and
\begin{equation}
 \rho(t):=\min \big( t/2, (\log t)\sqrt{t} \big)\, .
\end{equation}
The condition $|x_j -x_{j-1}| \le \rho(t_j - t_{j-1})$ turns out to be convenient for technical reasons but is not essential.
For the rest, as already mentioned, $X$ ressembles the term of order $q$ in the Taylor expansion in $\gb$ of the partition function ``restricted to a cell''.
We refer to \cite[Section 4.2]{cf:BL15} for a more elaborate discussion on the definition of $X(\go)$.

\medskip

We easily check that $\bbE[X(\go)]=0$ and
\begin{equation}\label{eq:isivar}
\bbE [(X(\go))^2]  =  \frac{1}{4 R^2 \ell^2 D(u)^{q}} \sum_{ \u{x} \in (\tilde\gL_0)^{q+1} ,\, \u{t}\in J_{\ell,u} }  P(\u{t},\u{x})^2  \leq 1\, .
 \end{equation}

\begin{lemma}
\label{lem:key}
With the choice of change of measure made in \eqref{def:g}-\eqref{def:X}, there exists some $\gb_\gep>0$ such that, for all $\gb\leq \gb_\gep$, one has
\[
\sum_{y\in\bbZ^2} \max_{x\in\gL_{0}} \bE_y\left[  \bbE^S \left[ g_{(1,0)}(\go) \right] \, ; \, S_{\ell} \in\Lambda_{y_1} \right]^{1/2} \leq  \frac{1}{4}\, .
\]
\end{lemma}

\section{Proof of the key Lemma \ref{lem:key}}
\label{sec:keylemma}

In this section, for notational convenience, we write $g(\go)$ instead of $g_{(1,0)}(\go)$.

\smallskip
First, if $A$ is chosen sufficiently large and $\| y \|_2\geq A$ ($\| \cdot \|_2$ denotes the Euclidean norm), then uniformly for $x\in \Lambda_0 = (-\sqrt \ell/2 ,\sqrt\ell/2]^2$, we have
\[
\bP_x(S_{\ell}\in \Lambda_{y}) \leq e^{- \frac14 \|y\|_2^2} \, .
\]
Therefore, as $g(\go)\le 1$, we have
\begin{align*}
\sum_{\|y \|_2 \geq A} \max_{x\in\gL_{0}} \bE_x\left[  \bbE^S \left[ g(\go) \right] \, ; \, S_{\ell} \in\Lambda_{y} \right]^{1/2}
&\leq \sum_{\|y\|_2 \geq A} \max_{x\in\gL_{0}}  \bP_x(S_{\ell}\in \Lambda_{y})^{1/2} \\
& \leq \sum_{\|y\|_2 \geq A} e^{- \frac14 \|y\|_2 ^2} \leq \frac18 \, ,
\end{align*}
where the last inequality holds provides that $A$ is large enough.
For the remaining sum, we use the (rather rough) bound
\begin{equation*}
\sum_{\|y \|_2  \leq A} \max_{x\in\gL_{0}} \bE_x\left[  \bbE^S \left[ g(\go) \right] \, ; \, S_{\ell} \in\Lambda_{y} \right]^{1/2} \leq 4 A^2 \max_{x\in\gL_{0}} \bE_x\left[  \bbE^S \left[ g(\go) \right]  \right]^{1/2}.
\end{equation*}
Therefore, we need to control $\bE_x\left[  \bbE^S \left[ g(\go) \right]  \right]$ for every $x\in \Lambda_0$:

\begin{lemma}
\label{lem:EyESg}
For any $\eta>0$, there exist constants $K(\eta)>0$ and $\gb_{0}(\gep,\eta)$ such that for all $\gb\leq \gb_0$, and for any $x\in \Lambda_0$
\begin{equation}
\bE_x\left[  \bbE^S \left[ g(\go) \right]  \right] \leq \eta \, .
\end{equation}
\end{lemma}

Applying this lemma with $\eta = \big( \frac{1}{32 A^2}\big)^2$, we have
\[\sum_{\|y \|_2  \leq A} \max_{x\in\gL_{0}} \bE_x\left[  \bbE^S \left[ g(\go) \right] \, ; \, S_{\ell} \in\Lambda_{y} \right]^{1/2} \leq \frac18 \, ,\]
and Lemma \ref{lem:key} is proven.
To prove Lemma \ref{lem:EyESg} we need some control over the distribution of $X(\go)$ under $\bbP^S$.
\begin{lemma}
\label{lem:EX}
For any $\gd>0$, there exist $R(\delta)$ and $\gb_0(\gep,\gd)$ such that, for every $x\in\gL_0$ and any $\gb\leq \gb_0$
\[ \bP_{x}\big( \bbE^S[X] \geq (1+\gep^2 )^q \big)  \geq 1-\gd \, .\]
\end{lemma}

\begin{lemma}
\label{lem:VarX}
For any $\gd>0$, and $R>0$, there exists some $\gb_0(\gep,\gd)$ such that,
for every $x\in\gL_0$ and any $\gb\leq \gb_0$
\begin{equation}
\bP_x\left[ \Var^S(X) \leq (1+\gep^3)^q \right]\ge 1-\delta, .
\end{equation}
\end{lemma}

\begin{proof}[Proof of Lemma \ref{lem:EyESg}]
Recalling the definition \eqref{def:g} of $g(\go)$, we have for any $S$,
\begin{equation}
\label{controlEg1}
\bbE^S \left[ g(\go) \right]  \leq  e^{-K} +  \bbP^S\left(X(\go) \leq e^{K^2}\right),
\end{equation}
and we choose $K$ large such that $e^{-K}\le \eta/6$.
We define the event
\[
\cA:= \big\{ \bbE^S[X] \geq (1+\gep^2 )^q \big\} \, \cap \,  \big\{ \Var^S(X) \leq (1+\gep^3)^q \big\},
\]
and we  apply Lemmas \ref{lem:EX} and \ref{lem:VarX} for $\delta=\eta/3$, so that $\bP_x(\cA)\geq 1-2\eta/3$, for every $x\in\gL_0$.
Then, choosing $\gb$ small enough (so that $q$ is sufficiently large) we have  
$e^{K^2} \le \tfrac{1}{2} (1+\gep^2 )^q $, so that on the event $\cA$,
Chebychev's inequality yields
\begin{equation}
 \bbP^S\left(X(\go) \leq e^{K^2}\right) \le \bbP^S \left( X-\bbE^S[X] \leq - \tfrac12  (1+\gep^2 )^q \right)   
\le \frac{4(1+\gep^3)^q}{(1+\gep^2)^{2q}} \le \eta/6.
\end{equation}
Hence from \eqref{controlEg1}, we have (still on the event $\cA$)
\begin{equation}
\bbE^S \left[ g(\go) \right]\le \eta/3.
\end{equation}
Using the bound $\bbE^S \left[ g(\go) \right] \le 1$ on the complement of $\cA$ (which has probability at most $2\eta/3$), we conclude the proof of Lemma~\ref{lem:EyESg}.
\end{proof}

\subsection{Proof of Lemma \ref{lem:EX}}

From the definition \eqref{def:X} of $X$, and recalling \eqref{newExpVar}, we have, for any trajectory of $S$
\begin{multline}
\label{eq:coco0}
\bbE^S[X] =  \frac{(\gl'(\gb))^{q+1}}{2 R\ell D(u)^{q/2}} \sum_{\u{t} \in J_{\ell,u}} P(\u{t},\u{S}^{(\u{t})})\, \ind_{\{S_{t_k} \in \tilde\gL_0 \, \forall k\in\{0,\ldots, q\}\}} \,
\\
\ge \left[ \frac{(\gl'(\gb))^{q+1}}{2R\ell D(u)^{q/2}} \sum_{\u{t} \in J_{\ell,u}} P(\u{t},\u{S}^{(\u{t})}) \right] \ind_{\big\{ \max\limits_{t\leq \ell} \| S_t \|_{\infty}\leq R \sqrt{\ell} \big\}} \, ,
\end{multline}
where we used the notation $$\u{S}^{(\u{t})}:=(S_{t_0}, S_{t_1},\ldots, S_{t_q}).$$
Note that if $R=R(\gd)$ is chosen sufficiently large, we have for all $x\in (-\sqrt{\ell}/2, \sqrt{\ell}/2]^2\cap\bbZ^2$
\begin{equation}\label{eq:coco}
\bP_x \left( \max_{t\leq \ell} \| S_t \|_{\infty}> R \sqrt{\ell} \right)\le \bP_0 \left( \max_{t\leq \ell} \| S_t \|_{\infty}> (R-1/2) \sqrt{\ell} \right)\le \delta/2\, .
\end{equation}
On the event $\{ \max_{t\leq \ell} \| S_t \|_{\infty}\leq R \sqrt{\ell}\}$, 
we use  \eqref{eq:boundsonevar} which gives $\gl'(\gb) \geq (1-\gep^3) \gb$, to obtain
\begin{align}
\label{eq:coco2}
\bbE^S[X]  &\geq \gb (1-\gep^3)^{q+1} (\gb^2 D(u))^{q/2} \ \frac{1}{2R\ell D(u)^q}  \sum_{\u{t} \in J_{\ell,u}} P(\u{t},\u{S}^{(\u{t})}) \notag \\
& \geq   (1+\gep^2)^{2q} \frac{1}{\ell D(u)^q}  \sum_{\u{t} \in J_{\ell,u}} P(\u{t},\u{S}^{(\u{t})})\, ,
\end{align}
where in the last line, we used \eqref{eq:compar}, and the inequality
\[
\frac{\gb}{2R} (1-\gep^3)^{q+1} (1+\gep)^{q/2} \geq  (1+\gep^2)^{2 q}.
\]
which is valid provided $\gep<1/10$ and $\gb$ is small enough.
Hence combining \eqref{eq:coco0} with \eqref{eq:coco}-\eqref{eq:coco2} we have 
\begin{equation}
 \bP_x\Big(  \bbE^S[X]  \geq (1+\gep^2)^{q}      \Big) \le 
 \delta/2+  \bP_x\left[\frac{1}{\ell D(u)^q}  \sum_{\u{t} \in J_{\ell,u}} P(\u{t},\u{S}^{(\u{t})})\ge \frac{1}{(1+\gep^2)^{q}} \right].
\end{equation}
Then, we obtain again a lower bound if we restrict the sum to $\u{t}$ such that $t_0 \leq \ell/2$. 
We set
\[J'_{\ell,u}:= \big\{ \u{t}=(t_0,\ldots,t_q) \, ; \, t_0\leq \ell/2 \, ; \, t_j-t_{j-1} \in (0,u] \ \forall j\in\{1,\ldots,q\}\big\} .\]
Note that  $J'_{\ell,u}\subset J_{\ell,u}$ provided $\gb$ is small enough (because $\ell/2+qu\le \ell$).
Therefore, it is sufficient to show that
\begin{equation}
\label{eq:W}
\bP_x \left( W_{\ell} < \frac{1}{(1+\gep^2)^{q}} \right) \leq \gd/2 \, .
\end{equation}
where 
\begin{equation}
W_{\ell}:=\frac{1}{\ell D(u)^q}  \sum_{\u{t} \in J'_{\ell,u}} P(\u{t},\u{S}^{(\u{t})}) \, .
\end{equation}
Note that the law $W_{\ell}$ does not depend on the starting point $x$, hence for the rest of the proof
we replace $\bP_x$ by $\bP$.
We achieve the bound \eqref{eq:W} by controlling the first two moments of $W$ (we actually prove that $W_{\ell}$ converges in probability to $\frac12$ as $\ell\to\infty$).
We have
\begin{align}
\bE \left[ W_{\ell} \right] &= \frac{1}{\ell D(u)^q} \sum_{\u{t} \in J'_{\ell,u}} \bE\left[  P(\u{t},\u{S}^{(\u{t})}) \right].
\end{align}
Note that by the definition of $P$,  $\bE\left[  P(\u{t},\u{S}^{(\u{t})}) \right]$ is translation invariant and thus 
\begin{align}
 \sum_{\u{t} \in J'_{\ell,u}} \bE\left[ P(\u{t},\u{S}^{(\u{t})}) \right] &=\frac{\ell}{2} \sum_{\{ \u{t} \in J'_{\ell,u} \ |  \ t_0=1 \}} \bE\left[  P(\u{t},\u{S}^{(\u{t})}) \right] \notag\\
 & = \frac{\ell}{2}
 \left(\sum_{t=1}^u\bP\big( S_{2t}=0 \ ; \ |S_t|\le \rho(t) \big)    \right)^q=:  \frac{\ell}{2}\left(\hat D(u)\right)^q \, .
\end{align}
It is an easy exercise to show that the restriction $|S_t|\le \rho(t) = \min(t/2\, ,\, (\log t)\sqrt{t})$ has not much effect, and that (recall \eqref{def:D})
there exists a constant $C$ such that
\begin{equation}
 D(u)-C \le \hat D(u) \le D(u) \, .
\end{equation}
Hence, we conclude that
\begin{equation}
\label{EWconverge}
\frac{1}{2}\left(\frac{D(u)-C}{D(u)} \right)^q \le  \bE \left[ W_{\ell} \right] \le  \frac{1}{2}.
\end{equation}
and $\bE \left[ W_{\ell} \right]$ converges to $1/2$ when $\gb$ tends to zero (recall \eqref{eq:defq}-\eqref{eq:frameu}).

\medskip

Now let us estimate the variance of $W$.
We define, for $j\in \bbN$ 
\begin{equation}
Y_{j} = \frac{1}{D(u)^q} \sum_{\u{t} \in J'_{\ell,u}(j)} P(\u{t},\u{S}^{(\u{t})}) - \left(\frac{\hat D(u)}{D(u)}\right)^q
\end{equation}
where $J'_{\ell,u}(j) := \{ \u{t}\in J'_{\ell,u} \, ;\, t_0= j \}$. We have $\bE[Y_j]=0$, and $W_{\ell}- \bE[W_{\ell}] = \frac{1}{\ell} \sum_{j=1}^{\ell/2} Y_j$.
Hence, 
\begin{equation}\label{eq:cov}
\mathbf{Var}( W_\ell^2 ) = \frac{1}{\ell^2 } \sum_{j_1, j_2=1}^{\ell/2} \bE\left[Y_{j_1}Y_{j_2}\right]\, ,
\end{equation}
and we can conclude by showing that most covariance terms are zero. More precisely we have 
\begin{lemma}
\label{lem:Y}
One has that
\begin{itemize}
\item[(i)] There exists a constant $C_1$ such that, for all $j$, with probability $1$, $|Y_j| \leq (C_1)^q$.
\item[(ii)] If $|j_1-j_2|\ge  uq$, then $\bE[Y_{j_1} Y_{j_2} ]=0$.
\end{itemize}
\end{lemma}

Replacing the terms in \eqref{eq:cov} by either zero (if $|j_1-j_2|\ge uq$) or $(C_1)^{2q}$ (in other cases)
we obtain that 
\begin{equation}
 \mathbf{Var}( W_\ell^2 )  \leq \frac{2 u q}{\ell} (C_1)^{2q} \leq 2q\, (C_1)^{2q}\,  \ell^{-\gep^2}.   
\end{equation}
Since $q$ and $(C_1)^{2q}$ grow slower than any power of $\ell$, $\mathbf{Var}( W_\ell^2 )$ tends to $0$ as $\ell$ goes to infinity (or $\gb\downarrow 0$).
As a consequence, Chebychev's inequality together with \eqref{EWconverge} gives that $W_{\ell}$ converges in probability to $\frac12$ as $\ell\to\infty$, and \eqref{eq:W} hence Lemma~\ref{lem:EX} are proven. \qed

\begin{proof}[Proof of Lemma \ref{lem:Y}]
For the first point, we remark that trivially $Y_j\ge -1$. For an upper bound,
by the local Central Limit Theorem for the simple random walk on $\bbZ^2$ (which can be obtained with little more than the application of Stirling formula), 
there exists a constant $c_1$ such that 
\begin{equation}\label{eq:lclsrw}
\forall x\in \bbZ^2, \, p(t,x) \leq \frac{c_1}{1+t}.
\end{equation}
We therefore have
\begin{align*}
\sum_{\u{t} \in J'_{\ell,u}(j)} P(\u{t},\u{S}^{(\u{t})})\leq  \left(\sum_{i=1}^{u}  \frac{c_1}{1+i} \right)^{q} \leq (c_1)^{q} (\log u )^{2q}.
\end{align*}
As $D(u)$ is also of order $\log u$ (cf.~\eqref{def:D}), we obtain the result for a suitable $C_1$.

\medskip

For $(ii)$, note that 
\[\bE[Y_{j_1} Y_{j_2}] = \sum_{z \in \bbZ^2} \bE[  Y_{j_1}  \ind_{\{S_{j_2} =z \}} Y_{j_2} ]  \, .\]
If $j_2>j_1+qu$, then conditionnally on $S_{j_2}=z$ $Y_{j_1}$ and $Y_{j_2}$ are independent, and 
$\bE[Y_{j_2}|S_{j_2} =z ] =0 $ for all $z\in\bbZ^2$. Hence the result.
\end{proof}

\subsection{Proof of Lemm \ref{lem:VarX}}

We are going to show a uniform bound on the variance which holds provided that $S$ satisfies
\begin{equation}\label{eq:hypo}
\max\big\{ |S_t-S_{t'}| \ / \ 1\le t\le t' \le \ell, \, |t-t'|\le qu  \big\} \le  (\log u)\sqrt u. 
\end{equation}
Note that if $\ell$ is sufficiently large (i.e.\ $\gb$ sufficiently small), for every $x\in\gL_0$
this occurs with $\bP_x$ probability larger than $1-\delta$, by standard properties of the simple random walk.

\medskip

For any trajectory $S$, we define a modified environment
\[\hat \go_{n,x} := \go_{n,x} - \gl'(\gb) \ind_{\{S_n=x\}} \, .\]
It is such that under $\bbP^S$, the variables $\hat \go_{n,x}$ are independent and centered, with a variance smaller than
$1+(\gep^3/2)$, see \eqref{newExpVar}-\eqref{eq:boundsonevar}.
We want to expand
\begin{equation}
\label{eq:var1}
\bbE^S [ X^2 ] =\frac{1}{4R^2\ell^2 D(u)^q} \bbE^S \left[  \left(  \sum_{ \u{x} \in (\tilde{\gL}_0)^{q+1} ,\, \u{t}\in J_{\ell,u} } P(\u{t}, \u{x}) \prod_{j=1}^q \Big( \hat\go_{t_j,x_j} + \gl'(\gb) \ind_{\{S_{t_j}=x_j\}}  \Big)\right)^2 \right]\, .
\end{equation}
We have
\begin{equation*}
\prod_{j=0}^q \Big( \hat\go_{t_j,x_j} + \gl'(\gb) \ind_{\{S_{t_j}=x_j\}}  \Big) = \sum_{r=0}^{q+1} (\gl'(\gb))^{r}  \sum_{A\subset \{0,\ldots , q\}, |A|=r} \prod_{k\in A} \ind_{\{S_{t_k}=x_k\}} \prod_{j\in  \{0,\ldots , q\}\setminus A}  \hat\go_{t_j,x_j} \, .
\end{equation*}
Therefore, when taking the square in \eqref{eq:var1}, one obtains a sum over $\u{x},\u{x}'\in (\tilde\gL_0)^{q+1}$,$\u{t},\u{t}' \in J_{\ell,u}$ of $P(\u{t},\u{x}) P(\u{t}',\u{x}')$ times 
\begin{equation*}
\sum_{r=0}^{q+1} \sum_{r'=0}^{q+1} (\gl'(\gb))^{r+r'} \sumtwo{A\subset \{0,\ldots , q\}, |A|=r}{B\subset \{0,\ldots , q\}, |B|=r'} \prodtwo{k\in A}{k'\in B }   \ind_{\{S_{t_k}=x_k\}}  \ind_{\{S_{t'_{k'}}=x'_{k'}\}}  \prodtwo{j\in  \{0,\ldots , q\}\setminus A}{j'\in  \{0,\ldots , q\}\setminus B}  \hat\go_{t_j,x_j}  \hat\go_{t'_{j'},x'_{j'}}   \, .
\end{equation*}
When taking the expectation under $\bbE^S$, the only non-zero terms are those with $r=r'$ and 
\begin{equation}
 \{ (t_j,x_j) \ | \ j\in \{0,\dots, q\} \setminus A\}:= \{ (t_{j'},x_{j'}) \ | \ j'\in \{0,\dots, q\} \setminus B\} \, .
\end{equation}
Note that the term $r=q+1$ corresponds exactly to $\bbE^S[X]^2$, and is therefore canceled when considering the variance.

\medskip
We need to introduce some additional notations to reorganize the sum: 
\[\cS_{m} := \big\{ \u{s} = (s_0,\ldots , s_{m}) \, ; \, 1\leq  s_0< \cdots< s_m  \leq \ell  \, ; \, s_{m}-s_1 \leq m u \big\} \, .\]
Also, for $r\geq 1$,  and any $\u{s}\in \cS_{q-r} $, a set of  $\u{s}$-\emph{compatible} $\u{t}$:
\[\cT_{r} (\u{s}) = 
\big\{  \u{t} = (t_1, \ldots , t_{r})   \, ;\, 1\leq  t_1 < \cdots< t_{r}  \leq \ell \, ;\, \u{s}\cdot\u{t}\in J_{\ell,u}  \big\} \, .\]
where $\u{s}\cdot\u{t}$ denotes the ordered sequence with $q+1$ elements which is obtained by reordering the values of the $s_i$'s and $t_j$'s.
Hence, isolating the term $r=0$, and recalling $\bbE^S[(\hat\go_{t,x})^2]\leq 1+(\gep^3/2) \leq 2$, we obtain
\begin{align}
\label{eq:var2}
\Var^S [ X ]  \leq \, \frac{(1+(\gep^3/2))^{q+1}}{4R^2\ell^2 D(u)^q} &\sum_{\u{x}\in (\tilde\gL_0)^{q+1} \, ,\, \u{t}\in J_{\ell,u}} P(\u{t},\u{x})^2  \notag\\
+  \frac{1}{4R^2\ell^2 D(u)^q} & \sum_{r=1}^{q} (\gl'(\gb))^{2r} 2^{q+1-r}   \notag\\
&\sum_{\u{s} \in \cS_{q-r} }  \sum_{\u{x}\in (\tilde\gL_0)^{q-r+1} }
 \sum_{\u{t}, \u{t}' \in \cT_{r}(\u{s})} P\big( (\u{s},\u{x}), (\u{t},  \u{S}^{(\u{t})} ) \big) P \big( (\u{s}, \u{x}) ,( \u{t}', \u{S}^{(t')}) \big) \, ,
\end{align}
where we used the notation
\[P \big((\u{s},\u{t}) , (\u{x}, \u{z}) \big)  = P(\u{s} \cdot \u{t} , \u{x} \cdot\u{z} ) ,\]
where $\u{x}\cdot \u{z}$ 
is defined (a bit improperly since the definition depends on $\u{s}$ and $\u{t}$) as
\begin{equation}
 (\u{x}\cdot\u{z})_k :=
 \begin{cases} x_i \text{ if } (\u{s} \cdot \u{t})_k=s_i,\\
                     z_j   \text{ if } (\u{s} \cdot \u{t})_k=t_j .
                   \end{cases}
 \end{equation}
The first term, according to \eqref{eq:isivar} is smaller than $(1+(\gep^3/2))^{q+1}$.
It remains to control the second term.
First, we restrict the summation over $x_0$ by showing that 
\begin{equation}
\label{eq:restrictx0}
 |x_0-S_{s_0}| \ge \sqrt{u} (\log u)^2 \quad \Rightarrow \quad  P \big( (\u{s}, \u{x}) ,( \u{t}, \u{S}^{(t)}) \big)=0.
\end{equation}
Indeed for $P$ to be positive, all coordinates $\u{x} \cdot\u{S}^{(\u{t})}$ must be within distance 
$q\rho(u)$ of one another.
However, \eqref{eq:hypo} implies that 
\begin{equation}
  |x_0-S_{t_0}|\ge  \sqrt{u} (\log u)^2-\sqrt{u} \log u > q\rho(u),
\end{equation}
provided that $\gb$ is small enough.
For the other values of $x_0$ we will make use of the following bound

\begin{lemma}\label{lem:checkthesum}
 There exists a constant $C_2$ such that for any realization of $S$ we have, for any $\u{s}\in \cS_{q-r}$ 
 \begin{equation}
 \sum_{\u{t} \in \cT_{r}(\u{s})} P\big( (\u{s},\u{x}), (\u{t},  \u{S}^{(\u{t})} ) \big)\le 2^q ( C_2 \log \ell)^r\prod_{i=1}^{q-r} p(s_i-s_{i-1},x_i-x_{i-1}).
\end{equation}
\end{lemma}

This implies, together with \eqref{eq:restrictx0}, that for any $S$ verifying \eqref{eq:hypo}
\begin{multline}
\label{bigsum1}
 \sum_{\u{s} \in \cS_{q-r} }  \sum_{\u{x}\in (\tilde\gL_0)^{q-r+1} } 
 \left( \sum_{\u{t}\in \cT_{r}(\u{s})} P\big( (\u{s},\u{x}), (\u{t},  \u{S}^{(\u{t})} ) \big) \right)^2 \\
 \le   4^q ( C_2 \log \ell)^{2r}  \sum_{\u{s} \in \cS_{q-r} }  \sumtwo{\u{x}\in (\tilde\gL_0)^{q-r+1} }{ |x_0-S_{s_0}| \le \sqrt{u}(\log u)^2} 
 \left( \prod_{i=1}^{q-r} p(s_i-s_{i-1},x_i-x_{i-1}) \right)^2.
\end{multline}
Summing over all possible $0\leq s_0\leq \ell$ and $x_0$, we obtain
\begin{multline}
\label{bigsum2}
 \sum_{\u{s} \in \cS_{q-r} } \sumtwo{\u{x}\in (\tilde\gL_0)^{q-r+1} }{ |x_0-S_{s_0}| \le \sqrt{u}(\log u)^2} 
 \prod_{i=1}^{q-r} p(s_i-s_{i-1},x_i-x_{i-1})^2 \\
 \le \ell \, 2 u (\log u)^4\times\bigg( \sumtwo{x\in \bbZ^2}{1\le t \le \ell} p(t,x)^2 \bigg)^{q-r}
 \le   2 \ell u (\log \ell)^{4 + q-r}\, ,
\end{multline}
where we used that $D(\ell)\leq \log \ell$ if $\ell$ is large enough, see \eqref{def:D}.
Hence, collecting \eqref{bigsum1}-\eqref{bigsum2}, the second term in \eqref{eq:var2} is bounded by
\begin{multline}\label{thisis}
  \frac{1}{4 R^2\ell^2 D(u)^q}  \sum_{r=1}^{q} (\gl'(\gb))^{2r} 2^{q+1-r} 
 \times 4^q (C_2 \log \ell)^{2r} \, 2 \ell u\, (\log \ell)^{4+q+r} \\
  \le  \frac{4^{2q} C_2^{2q}}{ R^2} \, \left(\frac{\log \ell}{D(u)} \right)^q  (\log \ell)^4  u \ell^{-1} \sum_{r=1}^{q}  \gb^{2r} (\log \ell)^{r} 
  \le  \frac{q (C_3)^q}{R^2} (\log \ell )^4 \ell^{-\gep^2} \, ,
\end{multline}
where we used that $\gl'(\gb)\leq 2\gb$ (see \eqref{eq:boundsonevar}) in the first inequality. In the second inequality, we used that $\gb^2 \log \ell \leq 2\pi$ from the definition of $\ell$ \eqref{def:ell}, that $D(u)\geq \frac14 \log \ell$ if $\ell$ is large enough (and $\gep<1/10$) so that $C_3:= 128\, C_2^2 \pi$ ;  we also used that $u\leq \ell^{1-\gep^2}$.
Since $q, (C_3)^q$ and $(\log \ell)^4$ grow slower than any power of $\ell$, the r.h.s.\ of \eqref{thisis} tends to zero as $\ell$ goes to infinity (or $\gb\to 0$).
Therefore, if $\gb$ is small enough, \eqref{eq:var2} implies
\begin{equation}
\Var^S [ X ]  \leq  (1+(\gep^3/2))^{q+1} + 1 \leq (1+\gep^3)^q\, .
\end{equation}
\qed

\subsection{Proof of Lemma \ref{lem:checkthesum}}

First of all, we divide the sum according to the way the $t$ coordinates are interlaced with the $s$ coordinates: for any given $\u{s}\in \cS_{q-r}$, we have
\begin{multline}
\label{mondieu}
 \sum_{\u{t} \in \cT_{r}(\u{s})} P\big( (\u{s},\u{x}), (\u{t},  \u{S}^{(\u{t})} ) \big) \leq  \sum_{0\le m_0< \ldots < m_{q-r} \leq r}\\
  \prod_{i=1}^{q-r} \sum_{s_{i-1}< t_{m_i+1}<\cdots<t_{m_{i+1}}< s_i} P\big( (s_{i-1},t_{m_i+1},\ldots, t_{m_{i+1}},s_i) , (x_{i-1}, S_{t_{m_i+1}},\cdots ,S_{t_{m_{i+1}}} , x_i)\big) \\
  \times \sum_{0< t_1<\cdots<t_{m_0}< s_0} P\big( (t_1,\ldots,t_{m_0}, s_0) , ( S_{t_1},\cdots ,S_{t_{m_0}} , x_0)\big) \\
  \times    \sum_{s_{q-r}< t_{m_{q-r}+1}<\cdots<t_{r}< \ell } P\big( (s_{q-r},t_{m_{q-r}},\ldots,t_r ) , (x_{q-r}, S_{t_{m_{q-r}}},\cdots ,S_{t_{r}} \big)\, .
\end{multline}
where we used the (rather unusual, but convenient) convention that if $m_{i+1}=m_i$ then
\begin{multline}\label{eq:conveda}
 \sum_{s_{i-1}< t_{m_i+1}<\cdots<t_{m_{i+1}}< s_i} P\big( (s_{i-1},t_{m_i+1},\ldots, t_{m_{i+1}},s_i) , (x_{i-1}, S_{t_{m_i+1}},\cdots ,S_{t_{m_{i+1}}} , x_i)\big)\\
 = P\big( (s_{i-1},s_i) , (x_{i-1},x_i)\big)=  p(s_i-s_{i-1},x_i-x_{i-1})\ind_{\{|x_i-x_{i-1}|\le \rho(s_i-s_{i-1})\}}
\end{multline}
Here $m_i$ designates the index of the last $t$-coordinate before $s_i$: there are $m_0$ $t$-coordinates before $s_0$, $m_{i}-m_{i-1}$ between $s_{i-1}$ and $s_{i}$,
and $r- m_{q-r}$ after $s_{q-r}$.  We also isolated the contribution of the $t$-coordinates smaller than $s_0$ and of those larger than $s_{q-r}$.
Note that there are $\binom{q+1}{q-r+1}\le 2^{q}$ possible interlacements $0\leq m_0\leq \cdots\leq m_{q-r} \leq r $, and we will bound the contribution of each of them separately.

\medskip
First, we deal with the contribution of the $t$ coordinates greater than $s_{q-r}$. We use \eqref{eq:lclsrw} to obtain
\begin{multline}
\label{tgreater}
\sum_{s_{q-r}< t_{m_{q-r}+1}<\cdots<t_{r}< \ell } P\big( (s_{q-r},t_{m_{q-r}+1},\ldots,t_r ) , (x_{q-r}, S_{t_{m_{q-r}}},\ldots ,S_{t_{r}} \big) \\
\leq \sum_{s_{q-r}< t_{m_{q-r}}<\cdots<t_{r}< \ell} \prod_{k=m_{q-r}+1}^{r} \frac{c_1}{1+t_k-t_{k-1}} \leq \big( c_1 \log \ell\big)^{r-m_{q-r}}\, .
\end{multline}
By symmetry, the same argument yields
\begin{equation}
\label{tsmaller}
\sum_{0 < t_1<\cdots<t_{m_0}< s_0} P\big( (t_1,\ldots,t_{m_0}, s_0) , ( S_{t_1},\cdots ,S_{t_{m_0}} , x_0)\big) \leq \big( c_1 \log \ell\big)^{m_0}\, .
\end{equation}
Then, we deal with the inner terms.
We have to prove that there exists a constant $C_2$ such that, 
for any $1\leq s \leq \ell$, $x\in\bbZ^2$, and any sequence $(V_n)_{n\geq 0}$ with $V_0=0$ and $V_s=x$,
\begin{equation}
\label{eq:toprove}
\sum_{ 0 < t_{1}<\cdots<t_{k}< s}  P\big( (0,t_{1},\ldots, t_{k},s) , (V_0=0, V_{t_{1}},\cdots ,V_{t_{k}} , V_s=x)\big) \\
\leq (C_2 \log \ell)^{k} p(s,x) \, .
\end{equation}
This, combined with \eqref{tgreater}-\eqref{tsmaller} and plugged in \eqref{mondieu} completes the proof of Lemma \ref{lem:checkthesum}.
To prove \eqref{eq:toprove}, the main ingredient is the following inequality.
\begin{lemma}\label{eq:techlem}
There exists a constant $C_2$ such that for all $t\ge 1$, for all $x$ satisfying $|x|\le t/2$, and
for all $0\le n \le t$ 
 \begin{equation}
  \sup_{\{z\in \bbZ^2 \ / \ |z|\le n/2, |x-z|\le (t-n)/2\}} \frac{\bP\big( S_n=z ; S_t=x  \big)}{ \bP \big( S_t=x \big)}\le \frac{C_2}{1+\min(n, t-n)}. 
 \end{equation}
\end{lemma}
This is a standard but technical estimate (note that the assumption on $z$ is not needed but simplifies the proof).
We use Lemma \ref{eq:techlem} now to prove \eqref{eq:toprove} by induction, and postpone its proof to the end of the section. 

\medskip
Note that \eqref{eq:toprove} for $k=0$ follows from our convention \eqref{eq:conveda}. For $k\ge 1$ note that
\begin{multline}
\label{aftertechlem}
  p(s-t_k,V_s-V_{t_k})p(t_k-t_{k-1},V_{t_k}-V_{t_{k-1}})\ind_{\{|V_{t_k}-V_{t_{k-1}}|\le \frac{1}{2}|t_{k}-t_{k-1}| \ ; \ |V_s-V_{t_{k}}|\le \frac{1}{2}|s-t_{k}| \}}
  \\ \le
  \frac{C_2}{1+\min(s-t_k,t_k-t_{k-1})} p(s-t_{k-1},V_s-V_{t_{k-1}}) \ind_{\{|V_s-V_{t_{k-1}}|\le \frac{1}{2}|s-t_{k-1}| \}}.
\end{multline}
Using the convention that $t_0=0$, $t_{k+1}=s$, it therefore gives the following upper bound on \eqref{eq:toprove}
\begin{multline}
 \sum_{t_0=0<t_1<\dots<t_k<s=t_{k+1}}
\prod_{i=0}^{k} p(t_{i+1}-t_i,V_{t_i}-V_{t_{i-1}})\ind_{\{|V_{t_i}-V_{t_{i-1}}|\le \frac{1}{2}|t_{i}-t_{i-1}|\}}\\
\le  2 C_2 \log(s)  \sum_{0<t_1<\dots<t_{k-1}<s}
\prod_{i=0}^{k-1} p(t_{i+1}-t_i,V_{t_i}-V_{t_{i-1}})\ind_{\{|V_{t_i}-V_{t_{i-1}}|\le \frac{1}{2}|t_{i}-t_{i-1}|\}}\, ,
\end{multline}
where we simply summed \eqref{aftertechlem} over $t_k$.
We can then conclude by induction.

\begin{proof}[Proof of Lemma \ref{eq:techlem}]
First, we simplify the problem thanks to a rotation: if we denote $S_n:=(X_n,Y_n)$, then letting 
$\tilde X_n:= X_n-Y_n$, $\tilde Y_n:= X_n+Y_n$, we obtain that $\tilde X_n$ and $\tilde Y_n$ are two independent symmetric nearest-neighbor random walks on $\bbZ$.
Then, writing $x=(x_1,x_2)$, one has that $\{S_t=x\} = \{\tilde X_{t} = x_1-x_2 \, ;\, \tilde Y_t = x_1+x_2 \}$.
Therefore, Lemma \ref{eq:techlem} reduces to a statement on the nearest-neighbor random walk on $\bbZ$: 
we only need prove that there exists a constant $c_2$ such that, for all $t\ge 2$, and all $\tilde x$ 
satisfying $|\tilde x|\leq t/2$, we have for all $1\leq n\leq t-1$,
 \begin{equation}
  \sup_{\{ \tilde z\in \bbZ \ / \ |\tilde z|\le n/2, |\tilde x- \tilde z|\le (t-n)/2 \}} \frac{\bP\big( \tilde X_n=\tilde z \big) \bP \big( \tilde X_{t-n}= \tilde x -\tilde z  \big)}{ \bP \big( \tilde X_t= \tilde x \big)}\le \frac{c_2}{\sqrt{\min(n, t-n)}}. 
 \end{equation}
This is equivalent to proving that for all $n$, $t$, and $k$ and $j$ satisfying 
\begin{equation}\label{eq:porra}
t/4 \le k\le 3t/4\, ,\quad
n/4 \le j \le 3n/4\, ,\quad
(t-n)/4 \le k-j \le 3(t-n)/4\, ,
\end{equation}
we have
\begin{equation}
 \frac{\binom{n}{j}\binom{t-n}{k-j}}{ \binom{t}{k}}\le  \frac{c_2}{\sqrt{\min(n, t-n)}}. 
\end{equation}
By symmetry, we only have to prove this inequality for $n\leq t/2$. 
Using Stirling's formula for all the binomial coefficients we obtain that there exists a constant $C>0$ such that 

\begin{multline}
  \frac{\binom{n}{j}\binom{t-n}{k-j}}{ \binom{t}{k}}\le C \sqrt{\frac{n(t-n)k(t-k)}{j(n-j)(k-j)(t-n-(k-j))t}} 
  \\ \times \left(\frac{n}{j}\right)^j\left(\frac{n}{n-j}\right)^{n-j}\left(\frac{t-n}{k-j}\right)^{k-j}\left(\frac{t-n}{t-n-(k-j)}\right)^{t-n-k+j}
  \left(\frac{k}{t}\right)^k  \left(\frac{t-k}{t}\right)^{t-k}
\end{multline}
Note that because of our assumptions \eqref{eq:porra} and $n\le t/2$, the square-root term is up to a multiplicative constant equivalent to $n^{-1/2}$.
Hence we just need to show that the factor on the second line is smaller than one.
If one considers $j$ as a continuous variable, elementary calculus implies that this term is maximized for $j=kn/t$, and that this maximal value is indeed $1$.
\end{proof}

\section{Upper bound}
\label{sec:upper}

By a superadditivity argument (see \cite[Proposition 2.5]{cf:CSY} and its proof) we have 
\begin{equation}
\label{superad}
 \gD\tf(\gb)\le -\frac{1}{N}\bbE \left[ \log \hat Z^{\gb,\go}_N\right] \, .
\end{equation}
We use this inequality for some the largest possible choice of $N:=N_{\gb,\gep}$ for which $\hat Z^{\gb,\go}_N$ is still of constant order (recall that it has mean one).

After a straightforward second moment computation, see Section~\ref{sec:secondmoment}, one realizes that the right choice is $N_{\gb,\gep}:=\exp\left((1-\gep)\pi \gb^{-2}\right)$, so that $\bbE\big[ \big(\hat Z^{\gb,\go}_{N_{\gb,\gep}} \big)^2\big]$ is uniformly bounded by a constant, see \eqref{secondmoment}.
This intuition is strengthened by the work (in preparation) of Caravenna, Sun and Zygouras \cite{cf:CSZ3} which proves that, when $\gb\to 0$, $\log \hat Z^{\gb,\go}_{N_{\gb,\gep}} $ converges in distribution towards a normal random variable whose expectation and variance depend only on $\gep$.

\smallskip
But to be able to use \eqref{superad}, we need to prove that $\log \hat Z^{\gb,\go}_{N_{\gb,\gep}}$ is concentrated around its mean. 
We stress that obtaining such a concentration result is not straightforward: standard techniques (e.g.\ using martingales)
give that the variance of $\log \hat Z^{\gb,\go}_N$ is bounded above by $CN$ (see \cite[Section 6]{cf:CSY}).
The reader can check that by using this bound as done in \cite[Section 7]{cf:L}, we would be off by a factor $\frac12$ for our bound on $\log \Delta \tf(\gb)$.
Obtaining better bounds for the variance of $\log \hat Z^{\gb,\go}_N$ is in general a very difficult problem and the best known general improvement 
are by $\log$ factors (see e.g.~\cite{cf:AZ}). 

However, in our context, the temperature depends on $N$ and we can use this fact to obtain sharper concentration results. 
To do so, we borrow ideas from \cite{cf:CTT}, and we obtain a uniform bound on the tail of $-\log \hat Z^{\gb,\go}_{N_{\gb,\gep}}$.

\subsection{Concentration of $\log \hat Z^{\gb,\go}_{N_{\gb,\gep}}$}

To simplify the exposition, we present the proof first in the case where $\go$ is bounded, and then quickly adapt it  to the general case with a suitable truncation procedure.

The boundedness is used to have a convex concentration inequality which does not depend on the number or variable considered.
Let us start with a convex concentration inequality for bounded variables.
It follows from \cite[Lemma 3.3]{cf:CTT} and a more usual concentration inequality \cite[Corrolary 4.10]{cf:Ledoux}

\begin{lemma}\label{lem:concconv}
There exists a constant $C_1$ such that for any $m\ge 0$,
for any  sequence of i.i.d.\ variables    $\eta= \left(\eta_1,\dots,\eta_m\right)$ satisfying
\begin{equation}\label{boundedness}
 \bbP\left( |\eta_1|<K \right)=1,
\end{equation}
and any convex set $A\subset \bbR^m$, we have
  \begin{equation}
 \bbP\left( \eta \in A\right) \bbP\left( d(\eta,A) >t \right) \le 2e^{-\frac{t^2}{ C_1 K^2}}.
 \end{equation}
\end{lemma}

We do not use the result above directly but as a tool to obtain a finer concentration result which is valid for function whose Lipschitz norm 
is controlled only a small set. It is a convex version of \cite[Proposition 1.6]{cf:Ledoux} (we refer to \cite{cf:CTT} for the details). 
If $f$ is a function of $\eta$ we let  $|\nabla f (\eta)|$ denote the  Euclidean norm of the gradient of $f$, 
\begin{equation}
|\nabla f (\eta) | =  \sqrt{ \sum^m_{i=1} \left(\frac{\partial f}{\partial \eta_{i}}(\eta)\right)^2}.
\end{equation}

\begin{proposition}[\cite{cf:CTT}, Proposition 3.4]
\label{thm:lapropo}
 Let $f$ be a convex function, and $\eta$ satisfy \eqref{boundedness}.
 Then for any $a$, $c$ and $t$ we have 
 \begin{equation}
  \bbP \Big( f(\eta) \geq a \, ; \, | \nabla f (\go) | \leq M \Big) \bbP\left( f(\eta)\le a-t \right)\le 2e^{-\frac{t^2}{ C_1 K^2 M^2}}
 \end{equation}
where the constant $C_1$ is identical to that of Lemma \ref{lem:concconv}.
 \end{proposition}

We want to apply this result to $\log  \hat Z^{\gb,\go}_{N_{\gb,\gep}}$, which is a convex function of $(\go_{n,x})_{1\le n\le N, |x|\le N}$.
To use Proposition \ref{thm:lapropo} efficiently we have to obtain a good bound on the norm of its gradient
\begin{equation}
\left|\nabla \log  \hat Z^{\gb,\go}_{N_{\gb,\gep}}\right|:=   \sqrt{ \sum^{N}_{n=1} \sum_{|x|\le n}
\left(\frac{\partial }{\partial \go_{x,n}}\log  \hat Z^{\gb,\go}_{N_{\gb,\gep}}\right)^2}\, .
\end{equation}
We prove the following in Section \ref{sec:secondmoment}.
\begin{lemma} \label{lem:keystatement}
 There exists positive constants $\gb_\gep$ and $M=M_{\gep}$ such that for all $\gb<\gb_{\gep}$ we have 
\begin{equation}
\label{concentre}
\bbP \Big(\hat Z^{\gb,\go}_{N_{\gb,\gep}} \ge 1/2  \, ; \, \left|\nabla \log  \hat Z^{\gb,\go}_{N_{\gb,\gep}}\right| \leq M \Big) \geq \gep/80 \, .
\end{equation}
 \end{lemma}

In the case where the environment $\go$ satisfies $|\go_{n,x}|\le K$ almost surely, we can use Proposition~\ref{thm:lapropo} directly with $a=-\log 2$, combined with Lemma \ref{lem:keystatement}. 
This yields 
\begin{equation}
 \bbP\left( \log \hat Z^{\gb,\go}_{N_{\gb,\gep}} \le -\log 2 -t \right) \le \frac{160}{\gep}e^{-\frac{t^2}{ C_1 K^2 M^2}},
\end{equation}
and hence that 
\begin{equation}
 \Delta\tf(\gb)\le - \frac{1}{N_{\gb,\gep}}\bbE \left[ \log \hat Z^{\gb,\go}_{N_{\gb,\gep}} \right]\le  \frac{C(\gep,M,K)}{N_{\gb,\gep}}.
\end{equation}

\smallskip
In the case where the environment is unbounded, we can deduce from \eqref{expomom}  that there exists $c_0<c$ such that
\begin{equation}
\forall v>0,\quad \bbP(|\go_1|\geq v) \leq 2 e^{-c_0 v},
\end{equation}
and hence
\begin{equation}
 \bbP \Big( \max_{n\le N, |x|\le n} |\go_{x,n}| \ge v \Big)\le 8 N^3 e^{-c_0 v}.
\end{equation}
From this we deduce two bounds. The first one is rough, but valid for any value of $\gb$ and $N$, and we use it in desperate cases
\begin{equation}\label{eq:desperate}
 \bbP\left[ \log \hat Z^{\gb,\go}_{N} \le -\left(\gb v+\gl(\gb)\right) N \right] \le   8 N^3 e^{-c_0 v}.
\end{equation}
The other one makes use of Proposition \ref{thm:lapropo}, that we apply to
\begin{equation}
 \tilde \go \mapsto f(\tilde \go):= \log  \bE\left[ \exp\left( \sum_{n=1}^N \gb \tilde \go_{n,S_n}-\gl(\gb)\right)\right]\, ,
\end{equation}
where
$\tilde \go_{x,n}= \go_{x,n}\ind_{\{|\go_{x,n}|\le (\log N)^2\}}$.
For $\gb$ small enough, Lemma \ref{lem:keystatement} gives that
\begin{multline}
 \bbP\big(  f(\tilde \go)\ge -\log 2,  \left|\nabla f(\tilde \go)\right| \leq M \big) \\ \ge 
 \bbP \Big(\hat Z^{\gb,\go}_{N_{\gb,\gep}} \ge 1/2  \, ; \, \left|\nabla \log  \hat Z^{\gb,\go}_{N_{\gb,\gep}}\right| \leq M \Big)
 - \bbP(\go \ne \tilde \go)\ge \frac{\gep}{160},
\end{multline}
where in $\bbP(\go \ne \tilde \go)$ we implicitly considered environments restricted to 
$n\in [1,N]$, $|x|\le N$, and in the last inequality we used that
$$ \bbP ( \go \ne \tilde \go )\le 8  N^3 e^{-c_0 (\log N)^2}.$$

Therefore, applying Proposition \eqref{thm:lapropo} to the function $f(\tilde \go)$,
we finally obtain 
 \begin{multline}\label{eq:notbounded}
 \bbP\left( \log \hat Z^{\gb,\go}_{N_{\gb,\gep}} \le -\log 2 -t \right) \\
 \le
  \bbP\left( f(\tilde \go) \le -\log 2 -t \right)+ \bbP(\go \ne \tilde \go)
  \le \frac{320}{\gep}e^{-\frac{t^2}{ C_1 (\log N)^4 M^2}}+
8  N^3 e^{-c_0 (\log N)^2}.
\end{multline}
\smallskip

In the end, combining \eqref{eq:notbounded} for ``small'' values of $t$ (e.g.\ $t\le N^2$) and \eqref{eq:desperate} for all other values, we conclude that 
\begin{equation}
  \Delta\tf(\gb)\le - \frac{1}{N_{\gb,\gep}}\bbE \left[ \log \hat Z^{\gb,\go}_{N_{\gb,\gep}} \right]\le \frac{C(\log N_{\gb,\gep})^2}{N_{\gb,\gep}} \, .
\end{equation}
This yields the result thanks to the definition of $N_{\gb,\gep}$.

\subsection{Second moment estimate, proof of Lemma \ref{lem:keystatement}}
\label{sec:secondmoment}

Let us set $\gamma(\gb):=\gl(2\gb)-2\gl(\gb)$, and note that 
\begin{equation}
\label{computesecond}
\bbE\left[ \left(\hat Z^{\gb,\go}_N\right)^2\right]=\bE^{\otimes 2} \left[ \exp\left( \gamma(\gb) \sum_{n=1}^N \ind_{\{S_n^{(1)} = S_n^{(2)}\}}\right)\right].
\end{equation}
Using \cite[Lemma 6.4]{cf:BL15}, for $\gep>0$ and $\gb$ sufficiently small, and choosing $N=N_{\gb,\gep}:=\exp\left((1-\gep)\pi \gb^{-2}\right)$, we have that
\begin{equation}
\label{secondmoment}
 \bbE\left[ \left(\hat Z^{\gb,\go}_{N_{\gb,\gep}}\right)^2\right]\le \frac{10}{\gep}.
\end{equation}
Therefore, thanks to the Paley-Zygmund inequality, we obtain  
\begin{equation}\label{okidoki}
 \bbP\left[ \hat Z^{\gb,\go}_{N_{\gb,\gep}}\ge 1/2\right]\ge \frac{1}{4 \bbE\left[ \left(\hat Z^{\gb,\go}_{N_{\gb,\gep}}\right)^2\right]}\ge \frac{\gep}{40}. 
\end{equation}

\smallskip
For notational simplicity let us write $f(\go):= \log \hat Z^{\gb,\go}_{N_{\gb,\gep}}.$
With \eqref{okidoki}, we have
\begin{align}
\bbP \Big( \hat Z^{\gb,\go}_{N_{\gb,\gep}}\ge \frac12 \, ; \, | \nabla f (\go) | \leq M \Big) & =  \bbP \Big( \hat Z^{\gb,\go}_{N_{\gb,\gep}}\ge \frac12 \Big) -  \bbP \Big( \hat Z^{\gb,\go}_{N_{\gb,\gep}}\ge \frac12 \, ; \, | \nabla f (\go) | > M \Big) \notag\\
\label{pourconcentre} &\geq \frac{\gep}{40} - \frac{1}{M^2} \bbE \Big[  | \nabla f (\go) |^2 \, \ind_{\{\hat Z^{\gb,\go}_{N_{\gb,\gep}} \geq 1/2\}}  \Big].
\end{align}
Then, a straightforward calculation gives
 \begin{equation}
 \label{gradient}
|\nabla f(\go)|^2 = \frac{\gb^{2}}{(\hat Z^{\gb,\go}_{N_{\gb,\gep}})^2} \bE^{\otimes 2}\left[ \sum_{n=1}^{N_{\gb,\gep}} \ind_{\{S_n^{(1)} = S_n^{(2)}\}} \, \exp\left( \sum_{n=1}^{N_{\gb,\gep}} \left( \gb ( \go_{n,S_n^{(1)}} + \go_{n,S_n^{(2)}})-2\gl(\gb)  \right) \right) \right] \, ,
 \end{equation}
so that, similarly to \eqref{computesecond}, we get
\begin{equation}
\bbE \Big[  | \nabla f (\go) |^2 \, \ind_{\{\hat Z^{\gb,\go}_{N_{\gb,\gep}} \geq 1/2\}}  \Big] \leq 4 \bE^{\otimes 2} \left[ \gb^2 \Big( \sum_{n=1}^{N_{\gb,\gep}} \ind_{\{S_n^{(1)} = S_n^{(2)}\}} \Big)  \exp\left( \gamma(\gb) \sum_{n=1}^{N_{\gb,\gep}} \ind_{\{S_n^{(1)} = S_n^{(2)}\}}\right)\right]\, .
\end{equation}

It is now standard to show that this last term is uniformly bounded for $\gb\leq \gb_{\gep}$, as done for example in \cite[\S~6.3]{cf:BL15}.
First, notice that $\gamma(\gb)\sim \gb^2$ as $\gb\downarrow 0$. Therefore if $\gb$ is small enough, we have that $\gamma(\gb)\leq (1+\gep^2/2) \gb^2$, and there exists a constant $C_{\gep}>0$ such that for all $\gb\leq \gb_{\gep}$ and all $N\geq 1$

\[\gb^2 \Big( \sum_{n=1}^{N} \ind_{\{S_n^{(1)} = S_n^{(2)}\}} \Big)  \exp\left( \gamma(\gb) \sum_{n=1}^{N} \ind_{\{S_n^{(1)} = S_n^{(2)}\}}\right)\leq C_{\gep} \exp\left( (1+\gep^2)\gb^2 \sum_{n=1}^{N} \ind_{\{S_n^{(1)} = S_n^{(2)}\}}\right). \]

Hence, exactly as in Section 6.3 of \cite{cf:BL15}, the term we need to bound is $\mathbf{Z}_{N}^{u}$, the partition function of a homogeneous pinning model with parameter $u=(1+\gep^2)\gb^2$ and underlying renewal $\tau =\{n \, ;\, S_n^{(1)}= S_n^{(2)}\}$. Referring to \cite{cf:BL15} (in particular Equations (6.24)-(6.31)), we have that $\mathbf{Z}_{N_{\gb,\gep}}^{u} \leq 10/\gep$ if $\gb$ is small enough, and we get that
\[\bbE \Big[  | \nabla f (\go) |^2 \, \ind_{\{\hat Z^{\gb,\go}_{N_{\gb,\gep}} \geq 1/2\}}  \Big] \leq 40\,  C_{\gep} /\gep \, . \]
In the end, choosing $M= \frac{80 \sqrt{C_{\gep}}}{\gep}$ in \eqref{pourconcentre} yields \eqref{concentre}.

\bigskip

\textbf{Acknowledgments:} The authors are grateful to Francesco Caravenna, Rongfeng Sun and Nikos Zygouras for communicating their result \cite{cf:CSZ3}, and thank in particular Francesco Caravenna for pointing out the techniques used in \cite{cf:CTT}.

\bibliographystyle{plain}

\end{document}